\documentclass{article}
\usepackage[utf8]{inputenc}
\usepackage{amsmath}
\usepackage{amssymb,amsthm}
\usepackage{bbm}
\usepackage{mathtools}
\usepackage{stmaryrd}

\usepackage[margin=1.5in]{geometry}
\usepackage{enumitem}

\usepackage[caption=false]{subfig}

\usepackage{pgf}
\usepackage{tikz}
\usepackage[utf8]{inputenc}
\usetikzlibrary{arrows,automata}
\usetikzlibrary{positioning}
\usetikzlibrary{decorations.pathmorphing}
\usetikzlibrary{shapes.geometric}

\usepackage{color}

\newtheorem{theorem}{Theorem}
\newtheorem{lem}{Lemma}

\theoremstyle{remark}

\theoremstyle{definition}

\newtheorem{example}{Example}

\newcommand{\mbf}{\mathbf}
\newcommand{\mbb}{\mathbb}
\newcommand{\mc}{\mathcal}

\newcommand\blfootnote[1]{%
  \begingroup
  \renewcommand\thefootnote{}\footnote{#1}%
  \addtocounter{footnote}{-1}%
  \endgroup
}

\newcommand{\e}[1]{^{(#1)}}

\DeclareMathOperator{\ran}{\mathrm{ran}}
\DeclareMathOperator{\diam}{\mathrm{diam}}

\title{Uncertain Wiretap Channels and Secure Estimation}
\author{Moritz Wiese \and
Karl Henrik Johansson \and
Tobias J. Oechtering \and
Panos Papadimitratos \and
Henrik Sandberg \and
Mikael Skoglund}

\begin{document}

\maketitle

\begin{abstract}
	Uncertain wiretap channels are introduced. Their zero-error secrecy capacity is defined.  If the sensor-estimator channel is perfect, it is also calculated. Further properties are discussed. The problem of estimating a dynamical system with nonstochastic disturbances is studied where the sensor is connected to the estimator and an eavesdropper via an uncertain wiretap channel. The estimator should obtain a uniformly bounded estimation error whereas the eavesdropper's error should tend to infinity. It is proved that the system can be estimated securely if the zero-error capacity of the sensor-estimator channel is strictly larger than the logarithm of the system's unstable pole and the zero-error secrecy capacity of the uncertain wiretap channel is positive.
\end{abstract}

\section{Introduction}

If\blfootnote{All authors are with the ACCESS Linnaeus Centre, KTH Royal Institute of Technology, SE-10044 Stockholm, Sweden. \\E-mails: \{moritzw, kallej, oech, papadim, hsan, skoglund\}@kth.se} ``independent noise" is assumed for every time step, it tends to be considered as stochastic in information theory. In contrast to this, in robust control theory it is common to consider dynamical systems with nonstochastic disturbances. In order to give a unified framework for the latter case, Nair has proposed a ``nonstochastic information theory" \cite{Nnst}. The basic channel model in \cite{Nnst} is the newly introduced \textit{uncertain channel}, a rule which determines which channel input can generate which channel outputs \textit{without} weighting the possible outputs given the inputs. For finite alphabets, every uncertain channel thus corresponds to a 0-1-matrix obtained from a stochastic matrix by replacing every positive entry by 1. Thus, uncertain channels are natural objects in zero-error information theory. Nair also introduced an analog to mutual information which plays the same role for the zero-error capacity of uncertain channels as mutual information for the capacity of discrete memoryless channels. 

In \cite{Nnst}, Nair applied his nonstochastic information theory to the problem of estimating an unstable scalar dynamical system with nonstochastic disturbances at a remote location which obtains sensor data through an uncertain channel $\mbf T$. He showed that the estimation error can be bounded uniformly if the zero-error capacity $C_0(\mbf T)$ of $\mbf T$ is strictly larger than the logarithm of the system's unstable pole $\lambda>1$. This is ``almost" necessary as well in the sense that $C_0(\mbf T)\geq\log \lambda$ is required for uniform boundedness of the estimation error.

In this paper we add to the above problem that an eavesdropper overhears the communication between sensor and estimator via a second uncertain channel. The estimation error at the intended location should again be bounded uniformly, whereas for every eavesdropper output sequence there should be two system paths whose distance tends to infinity with increasing time. We call this the problem of secure estimation. A similar problem has been studied in \cite{LLZ} for the case of stochastic system and channel noise. 

For our nonstochastic setting, this leads to the introduction of the uncertain wiretap channel: a pair $(\mbf T_B,\mbf T_C)$ of uncertain channels with common input alphabet. A zero-error wiretap code is a zero-error code for $\mbf T_B$ such that every eavesdropper output word can be generated by at least two different messages. Surprisingly, positivity of the zero-error secrecy capacity $C_0(\mbf T_B,\mbf T_C)$ is sufficient, in addition to Nair's sufficient condition for a bounded estimation error, in order for secure estimation to be possible. The reason for this is that the system's instability helps to achieve the goal of security as soon as a sufficiently large error on the eavesdropper side has been introduced at the beginning of transmission.

The schemes for data transmission from the sensor to the estimator apply block codes. Thus there are inter-decoding times where no new data arrive at the estimator. The error at those times increases with communication delay. On the other hand, we show that the estimation error at decoding times can be made to vanish asymptotically at the cost of increased delay. Similarly, we provide a lower bound on the speed of divergence for the eavesdropper's error which increases with increasing delay. 

We calculate the secrecy capacity in the case of a perfect sensor-estimator channel. It either equals zero or the logarithm of the size of the input alphabet. An example shows that for general uncertain wiretap channels, no secure message transmission may be possible at blocklength 1, whereas a positive transmission rate is achieved for blocklengths $\geq 2$. It also shows that encoders for zero-error wiretap codes in general have to be strictly uncertain channels, i.~e.\ every message can be mapped to several possible codewords similar to the use of stochastic encoders for stochastic wiretap channels. We do not apply Nair's nonstochastic information-theoretic quantities in any of the analyses. Further, uncertain wiretap channels do not appear to provide new insights for the study of zero-error capacity, an overview of which is given in \cite{KO}.

\textit{Outline:} Section II describes the problems considered, Section III presents the results and Section IV contains the proofs.

\section{Model}

\subsection{Uncertain Channels}

Let $\mbf A,\mbf B$ be finite alphabets. An \textit{uncertain channel from $\mbf A$ to $\mbf B$} is a mapping $\mbf T:\mbf A\rightarrow2^{\mbf B}_*:=2^{\mbf B}\setminus\{\varnothing\}$. For any $a\in\mbf A$, the set $\mbf T(a)$ is the family of possible output values of the channel given the input $a$. Only one of the elements of $\mbf T(a)$ will actually be attained when transmitting $a$. That $\mbf T(a)\neq\varnothing$ for all $a$ means that every input generates an output. We will write $\ran(\mbf T)$ for the set of possible outputs of $\mbf T$, i.~e.\ $\ran(\mbf T)=\cup_{a\in\mbf A}\mbf T(a)$. 

An \textit{$M$-code} is a collection $\{\mbf F(m):1\leq m\leq M\}$ of nonempty and mutually disjoint subsets of $\mbf A$. This is equivalent to an uncertain channel $\mbf F:\{1,\ldots,M\}\rightarrow2^{\mbf A}_*$ with disjoint output sets, so we will often denote such a code by $\mbf F$. The necessity of codes with $\lvert\mbf F(m)\rvert\geq 2$ for some $m$ is shown in Examples \ref{ex:superact} and \ref{ex:singletonnec}. It is similar to the necessity of stochastic encoders for stochastic wiretap channels. 

Obviously, first applying $\mbf F$ and then $\mbf T$ leads to a new uncertain channel $\mbf T\circ\mbf F:\{1,\ldots,M\}\rightarrow2^{\mbf B}_*$ called the \textit{composition of $\mbf F$ and $\mbf T$}. Formally, we have for any $m\in\{1,\ldots,M\}$
\[
	(\mbf T\circ\mbf F)(m):=\mbf T(\mbf F(m)):=\bigcup_{a\in\mbf F(m)}\mbf T(a).
\]

A nonstochastic $M$-code $\mbf F$ is called a \textit{zero-error $M$-code for $\mbf T$} if  for any $m,m'\in\{1,\ldots,M\}$ with $m\neq m'$
\begin{equation}\label{eq:reliability}
	\mbf T(\mbf F(m))\cap\mbf T(\mbf F(m'))=\varnothing.
\end{equation}
Thus every possible channel output $y\in\ran(\mbf T\circ\mbf F)$ can be associated to a unique message $m$. For this to hold it is necessary that the sets $\mbf F(m)$ be disjoint, which is the reason for this assumption in the definition of $M$-codes.

Given an additional finite alphabet $\mbf C$, an \textit{uncertain wiretap channel} is a pair of uncertain channels $(\mbf T_B:\mbf A\rightarrow2^{\mbf B}_*,\mbf T_C:\mbf A\rightarrow2^{\mbf C}_*)$. The interpretation is that the outputs of channel $\mbf T_B$ are received by the intended receiver, whereas the outputs of $\mbf T_C$ are heard by an eavesdropper. See Fig.\ \ref{fig:singletonnec} for an example of an uncertain wiretap channel.

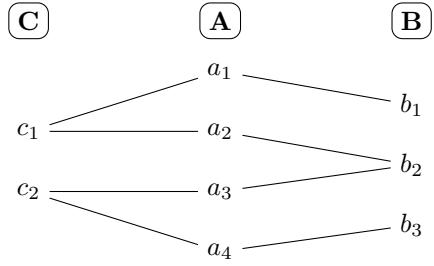
\begin{figure}
\centering
	\begin{tikzpicture}[alphabet/.style={draw, rounded corners}]
		\node[alphabet] (Eve) {$\mbf C$};
		\node[alphabet, right = 2cm of Eve] (Alice) {$\mbf A$};
		\node[alphabet, right = 2cm of Alice] (Bob) {$\mbf B$};
		
		\node[below = .2cm of Alice] (a1) {$a_1$};
		\node[below = 1cm of Alice] (a2) {$a_2$};
		\node[below = 1.8cm of Alice] (a3) {$a_3$};
		\node[below = 2.6cm of Alice] (a4) {$a_4$};
		
		\node[below = 1cm of Eve] (c1) {$c_1$};
		\node[below = 1.8cm of Eve] (c2) {$c_2$};
		
		\node[below = .6cm of Bob] (b1) {$b_1$};
		\node[below = 1.4cm of Bob] (b2) {$b_2$};
		\node[below = 2.2cm of Bob] (b3) {$b_3$};
		
		\draw (a1) -- (b1);
		\draw (a1) -- (c1);
		\draw (a2) -- (b2);
		\draw (a2) -- (c1);
		\draw (a3) -- (b2);
		\draw (a3) -- (c2);
		\draw (a4) -- (b3);
		\draw (a4) -- (c2);
	\end{tikzpicture}
	\caption{The channel from Example \ref{ex:singletonnec}. A line between $a_i$ and $b_j$ indicates that $b_j\in\mbf T_B(a_i)$, similar for $a_i$ and $c_j$.}\label{fig:singletonnec}
\end{figure}

An $M$-code $\mbf F$ is called a \textit{zero-error wiretap $M$-code for $(\mbf T_B,\mbf T_C)$} if it is a zero-error code for $\mbf T_B$ and additionally for every $c\in\ran(\mbf T_C\circ\mbf F)$ there are messages $m\neq m'$ such that
\begin{equation}\label{eq:security}
	c\in\mbf T_C(\mbf F(m))\cap\mbf T_C(\mbf F(m')).
\end{equation}
Thus every output $c\in\ran(\mbf T_C\circ\mbf F)$ can be generated by at least two possible messages.

We define the \textit{$n$-fold product of an uncertain channel} $\mbf T:\mbf A\rightarrow2^{\mbf B}_*$ as the uncertain channel $\mbf T^n:\mbf A^n\rightarrow(2^{\mbf B}_*)^n$ defined by
\begin{equation}\label{eq:nprod}
	\mbf T^n(a^n)=\mbf T(a_1)\times\cdots\times\mbf T(a_n).
\end{equation}
We call an $M$-code $\mbf F$ on the alphabet $\mbf A^n$ an \textit{$(M,n)$-code}. Given an uncertain channel $\mbf T$, an $(M,n)$-code $\mbf F$ is called a \textit{zero-error $(M,n)$-code for $\mbf T$} if \eqref{eq:reliability} is satisfied with $\mbf T\circ\mbf F$ replaced by $\mbf T^n\circ\mbf F$. We set $N_{\mbf T}(n)$ to be the maximal $M$ such that there exists a zero-error $(M,n)$-code for $\mbf T$ and define the \textit{zero-error capacity of $\mbf T$} by
\begin{equation}\label{eq:C0}
	C_0(\mbf T):=\sup_n\frac{\log N_{\mbf T}(n)}{n}.
\end{equation}
Given an uncertain wiretap channel $(\mbf T_B,\mbf T_C)$, an $(M,n)$-code $\mbf F$ is called a \textit{zero-error wiretap $(M,n)$-code for $(\mbf T_B,\mbf T_C)$} if it is a zero-error code for $\mbf T_B$ and if \eqref{eq:security} holds with $\mbf T_C\circ\mbf F$ replaced by $\mbf T_C^n\circ\mbf F$. We define $N_{(\mbf T_B,\mbf T_C)}(n)$ to be the maximal $M$ such that there exists a zero-error wiretap $(M,n)$-code for $(\mbf T_B,\mbf T_C)$. Then
\begin{align}
	C_0(\mbf T_B,\mbf T_C):=\sup_n\frac{\log N_{(\mbf T_B,\mbf T_C)}(n)}{n}\label{eq:C0WT}
\end{align}
is called the \textit{zero-error secrecy capacity of $(\mbf T_B,\mbf T_C)$}. Due to the superadditivity of the sequences $\log N_{\mbf T}(n)$ and $\log N_{(\mbf T_B,\mbf T_C)}(n)$, the suprema in \eqref{eq:C0} and \eqref{eq:C0WT} can be replaced by limits by the well-known Fekete's lemma \cite{Fek}, see also \cite{CK}.

\subsection{The Unstable Dynamical System}

Let $\lambda>1$ and consider the real-valued system
\begin{subequations}\label{eq:dynsyst}
\begin{align}
	x(t+1)&=\lambda x(t)+w(t),\label{eq:dynrec}\\
	x(0)&=0.
\end{align}
\end{subequations}
where $w(t)$ is a nonstochastic disturbance with range $[-\Omega/2,\Omega/2]$ for some $\Omega>0$. With
\begin{equation}\label{eq:finhorizbound}
	\tilde\Omega^{(t)}:=\frac{\Omega}{\lambda-1}(\lambda^t-1),
\end{equation}
the range of possible values of this system at time $t$ equals $[-\tilde\Omega^{(t)}/2,\tilde\Omega^{(t)}/2]$, whose diameter grows exponentially in $t$. A sensor performs perfect state measurements, encodes them and sends them through an uncertain wiretap channel $(\mbf T_B,\mbf T_C)$. The dynamic system and the channel are synchronous, i.~e.\ one symbol can be transmitted through the channel at every system time step. The goal is that the receiver of $\mbf T_B$ (the estimator) be able to estimate the state with bounded estimation error and the eavesdropper's estimation error tend to infinity.

Formally, a \textit{transmission scheme} $(n_k,f_k,\varphi_k)_{k=1}^\infty$ consists of a bounded sequence of positive natural numbers $(n_k)_{k=1}^\infty$ and, defining $t_k:=\sum_{i=1}^kn_i$, for each $k\in\mbb N$ an uncertain channel $f_k:\mbb R^{t_k}\rightarrow2^{\mbf X^{n_k}}_*$ and a mapping $\varphi_k:\mbf Y^{t_k}\rightarrow\mbb R^{n_k}$. Every uncertain channel $f_k$ maps the observations of the system state up till time $t_k$ into one of several possible codewords of length $n_k$. The receiver of $\mbf T_B$ uses $\varphi_k$ to produce from all symbols received so far an estimate $\hat x(t_k),\ldots,\hat x(t_{k+1}-1)$ of the system states $x(t_k),\ldots,x(t_{k+1}-1)$. 

The minimal delay which has to be tolerated is $\max_kn_k$. At this delay, the receiver has good estimates for the states at times $t_k$ but has to extrapolate for the states $x(t_k+1),\ldots,x(t_{k+1}-1)$. In particular, for the first $t_1-1$ steps of the evolution, the estimator has to rely on a rule which is independent of any observations and which we assume to estimate $\hat x(t)=0\;(0\leq t\leq t_1-1)$. Further, every system path $(x(t))_{t=0}^\infty$ generates a sequence $(c_t)_{t=1}^\infty$ of eavesdropper outputs.

Given a transmission scheme $(n_k,f_k,\varphi_k)_{k=1}^\infty$ and a sequence of estimates $(\hat x(t))_{t=0}^\infty$, we denote by $\mbf R_B((\hat x(t))_{t=0}^\infty)$ the set of system paths $(x(t))_{t=0}^\infty$ which using the transmission scheme can generate $(\hat x(t))_{t=0}^\infty$. One can consider $\mbf R_B$ as an uncertain channel in the reverse direction with $\mbb R^\infty$ as input and output alphabet. Similarly, for any infinite sequence $(c_t)_{t=1}^\infty\in\mbf Z^\infty$ of eavesdropper outputs, we denote by $\mbf R_C((c_t)_{t=1}^\infty)$ the set of system paths $(x(t))_{t=0}^\infty$ which can give rise to $(c_t)_{t=1}^\infty$.

For two sequences $(a_t)_{t=1}^\infty$ and $(b_t)_{t=1}^\infty$ let us define their distance to be $\lVert (a_t)-(b_t)\rVert_\infty:=\sup_t\lvert a_t-b_t\rvert$. For a set $S$ of sequences we define its diameter by
\[
	\diam(S):=\sup\{\lVert(a_t)-(b_t)\rVert_\infty:(a_t),(b_t)\in S\}.
\]
The transmission scheme $(n_k,f_k,\varphi_k)_{k=1}^\infty$ is called \textit{reliable } if the estimation error is bounded uniformly in the estimates, i.~e.\ there exists a constant $\kappa>0$ such that for every possible estimate sequence $(\hat x(t))_{t=0}^\infty$, 
	\[
		\sup\{\lVert(x(t))-(\hat x(t))\rVert_\infty:(x(t))_{t=0}^\infty\in\mbf R_B((\hat x(t))_{t=0}^\infty)\}\leq \kappa.
	\]
Further, $(n_k,f_k,\varphi_k)_{k=1}^\infty$ is called \textit{secure} if for every sequence $(c_t)_{t=1}^\infty\subset\mbf C^\infty$
	\[
		\diam(\mbf R_C((c_t)_{t=1}^\infty))=\infty.
	\]
Note that security is an asymptotic property due to the boundedness of the range of possible system states in any finite time horizon, cf.\ \eqref{eq:finhorizbound}. Upon receiving a sequence $(c_t)_{t=1}^\infty$ of channel outputs generated by a secure transmission scheme, the eavesdropper will not be able to estimate the system path $(x(t))_{t=0}^\infty$ that generated $(c_t)_{t=1}^\infty$ with a bounded estimation error.

\section{Results}

\subsection{Main Results}

\begin{theorem}\label{thm:main}
	A reliable and secure transmission scheme exists if $C_0(\mbf T_B)>\log\lambda$ and $C_0(\mbf T_B,\mbf T_C)>0$.
\end{theorem}

The main idea behind Theorem 1 is that the system's instability helps to achieve the goal of security as soon as a sufficiently large error on the eavesdropper side has been introduced at the beginning of transmission.

To apply Theorem \ref{thm:main}, $C_0(\mbf T_B)$ and $C_0(\mbf T_B,\mbf T_C)$ have to be known. However, the zero-error capacity $C_0(\mbf T_B)$ is unknown for most channels except a few special cases, cf.\ \cite{KO}. Neither do we provide a general formula for $C_0(\mbf T_B,\mbf T_C)$ here. A solution can be given, though, when the calculation of $C_0(\mbf T_B)$ is trivial.

\begin{theorem}\label{lem:TBinj}
	If $\mbf T_B$ is an injective function from $\mbf A$ to $\mbf B$, then $C_0(\mbf T_B,\mbf T_C)\in\{0,\log\lvert\mbf A\rvert\}$. Further, $C_0(\mbf T_B,\mbf T_C)=0$ if and only if there is no zero-error wiretap $(M,1)$-code for $(\mbf T_B,\mbf T_C)$ for any $M\geq 2$.
\end{theorem}

For the proof of Theorem 2, it is sufficient to consider codes with $\lvert\mbf F(m)\rvert=1$ for all $1\leq m\leq M$. The number of those elements of $\mbf A^n$ which cannot be used as codewords grows exponentially, at a rate which is less than $\log\lvert\mbf A\rvert$ if and only if there is no zero-error wiretap $(M,1)$ code for $(\mbf T_B,\mbf T_C)$ for any $M\geq 2$. Thus the number of elements of $\mbf A^n$ that \textit{can} be used either asymptotically grows with rate $\log\lvert\mbf A\rvert$ or equals 0.

\subsection{Estimation Error and Divergence Coefficient}

We study some additional properties of secure estimation schemes. As mentioned above, using a transmission scheme $(n_k,f_k,\varphi_k)_{k=1}^\infty$ with delay $\max_kn_k$, the estimates of system states $x(t)$ with $t\neq t_k\;(k\in\mbb N)$ have to be extrapolated from the last good estimate. Thus the estimation error after a decoding time $t_k$ grows exponentially until the next decoding time $t_{k+1}$. However, for any $\varepsilon>0$ the estimation error at times $(t_k)_{k=1}^\infty$ can be made smaller than $\varepsilon$ at least for large $k$ if the inter-decoding intervals $n_k\; (k\in\mbb N)$ (and thus the inter-decoding estimate errors) are sufficiently large:

\begin{lem}\label{lem:esterr}
	For every $\varepsilon>0$ there exists a transmission scheme such that for every sequence $(\hat x(t))_{t=0}^\infty$ of estimates and every $(x(t))_{t=0}^\infty\in\mbf R_B((\hat x(t))_{t=0}^\infty)$,
	\[
		\limsup_{k\rightarrow\infty}\lvert x(t_k)-\hat x(t_k)\rvert\leq\varepsilon.
	\]
	If $C_0(\mbf T_B,\mbf T_C)>\log\lambda$, then the limit superior can even be replaced by a supremum.
\end{lem}

Another parameter of interest is the speed of divergence of the diameter of the set of possible system states given eavesdropper outputs $(c_t)_{t=1}^T$ as $T\rightarrow\infty$. Given a zero-error wiretap $(M,n)$-code $\mbf F$, we define for every possible eavesdropper channel output $(c_t)_{t=1}^n\in\ran(\mbf T_C^n\circ\mbf F)$
\[
	\delta((c_t)_{t=1}^n)\\=\max\{\lvert m-m'\rvert+1:(c_t)_{t=1}^n\in\mbf T_C^n(\mbf F(m))\cap\mbf T_C^n(\mbf F(m'))\}.
\]
Clearly $2\leq\delta((c_t)_{t=1}^n)\leq M$. We then set
\[
	L:=\min\{\delta((c_t)_{t=1}^n):(c_t)_{t=1}^n\in\ran(\mbf T_C^n\circ\mbf F)\}
\]
and call $\mbf F$ a $(M,L,n)$-code. We also define
\[
	\Delta_{(\mbf T_B,\mbf T_C)}(n):=\max\left\{\frac{L-1}{M-1}:\mbf F\text{ is }(M,L,n)\text{-code}\right\}.
\]
Clearly, $0<\Delta_{(\mbf T_B,\mbf T_C)}(n)\leq 1$.

\begin{lem}\label{lem:divcoeff}
	For every $\varepsilon>0$ there exists a transmission scheme $(n_k,f_k,\varphi_k)_{k=1}^\infty$ such that for every eavesdropper output sequence $(c_t)_{t=1}^\infty$ there exist system paths $(x(t))_{t=1}^\infty,(x'(t))_{t=1}^\infty\in\mbf R_C((c_t)_{t=1}^\infty)$ satisfying
	\[
		\liminf_{T\rightarrow\infty}\frac{\lVert (x(t))_{t=1}^T-(x'(t))_{t=1}^T\rVert_\infty}{\lambda^T}\geq\frac{\Omega}{\lambda-1}\sup_n\Delta_{(\mbf T_B,\mbf T_C)}(n)-\varepsilon.
	\]
\end{lem}

The term on the right-hand side of the inequality in Lemma \ref{lem:divcoeff} is positive if $\varepsilon$ is chosen small enough. The case $\sup_n\Delta_{(\mbf T_B,\mbf T_C)}(n)=1$ corresponds to complete eavesdropper ignorance, cf.\ \eqref{eq:finhorizbound}.

\subsection{Uncertain Wiretap Channels}

We first note that the divergence coefficient increases with increasing blocklength (and hence delay). Thus we find a trade-off between the growth rate for the eavesdropper's estimation error and the delay:

\begin{lem}\label{lem:mittldist}
	If $C_0(\mbf T_B,\mbf T_C)>0$, then
	\[
		\sup_n\Delta_{(\mbf T_B,\mbf T_C)}(n)=\lim_{n\rightarrow\infty}\Delta_{(\mbf T_B,\mbf T_C)}(n)>0.
	\]
\end{lem}

Next we have a closer look at the zero-error secrecy capacity of uncertain wiretap channels. To study the zero-error capacity of an uncertain channel $\mbf T:\mbf A\rightarrow2^{\mbf B}_*$, one associates to it the following graph $G(\mbf T)$: its vertex set equals $\mbf A$ and an edge is drawn between $a,a'\in\mbf A$ if $\mbf T(a)\cap\mbf T(a')\neq\varnothing$. In that case we write $a\sim a'$.

The graph $G(\mbf T^n)$ corresponding to the $n$-fold product channel $\mbf T^n$ (see \eqref{eq:nprod}) is the \textit{strong $n$-fold product of $G(\mbf T)$} denoted by $G(\mbf T)^n$, in particular $G(\mbf T^n)=G(\mbf T)^n$. Here for any graph $G$ with vertex set $\mbf A$, the strong product $G^2$ of $G$ with itself is defined as follows: The vertex set of $G^2$ is $\mbf A^2$ and $(a_1,a_2)\sim(a_1',a_2')$ if 1) $a_1\sim a_1'$ and $a_2=a_2'$ or 2) $a_2\sim a_2'$ and $a_1=a_1'$ or 3) $a_1\sim a_1'$ and $a_2\sim a_2'$.

Finding the zero-error capacity of $\mbf T$ now amounts to finding the asymptotic behavior as $n\rightarrow\infty$ of the sizes of maximal independent systems of the graphs $G(\mbf T^n)$, cf.\ \cite{KO}. We define an \textit{independent system} in a graph as a set $\{\mbf F(1),\ldots,\mbf F(M)\}$ of mutually disjoint subsets of the vertex set $\mbf A$ such that no two vertices $a,a'$ belonging to different subsets $\mbf F(m)\neq\mbf F(m')$ are connected by an edge. 

To treat uncertain wiretap channels $(\mbf T_B,\mbf T_C)$, we consider a hypergraph structure $H(\mbf T_C^n)$ induced on $\mbf A^n$ in addition to the graph structure $G(\mbf T_B^n)$. A hypergraph consists of a vertex set together with a set of subsets, called \textit{hyperedges}, of this vertex set. The vertex set of $H(\mbf T_C^n)$ equals $\mbf A^n$. Every hyperedge is generated by a $(c_t)_{t=1}^n\in\mbf C^n$: we set $e((c_t)_{t=1}^n):=\{(a_t)_{t=1}^n\in\mbf A^n:(c_t)_{t=1}^n\in\mbf T_C^n((a_t)_{t=1}^n)\}$.

It is easy to see that $H(\mbf T_C^n)$ is the \textit{$n$-fold square product} $H(\mbf T_C)^n$, cf.\ \cite{HOShyperprod}. For any hypergraph $H$ with vertex set $\mbf A$ and hyperedge set $\mc E\subset 2^{\mbf A}$, the square product $H^2$ of $H$ with itself is defined as follows: The vertex set of $H^2$ is $\mbf A^2$ and the hyperedge set equals $\mc E^2:=\{e\times e':e,e'\in\mc E\}$. 

A zero-error wiretap $(M,n)$-code $\mbf F$ then is nothing but a collection of disjoint subsets $\{\mbf F(1),\ldots,\mbf F(M)\}$ of $\mbf A^n$ satisfying the two following properties:
\begin{enumerate}
	\item It is an independent system for $G(\mbf T_B^n)$;
	\item For every hyperedge $e$ of $H(\mbf T_C^n)$ there exist at least two different $m,m'$ such that $e$ has nonempty intersection with both $\mbf F(m)$ and $\mbf F(m')$. 
\end{enumerate}

This (hyper-)graph theoretic language is applied in the proof of Theorem \ref{lem:TBinj}. The following very interesting example gives additional insight into the nature of general uncertain wiretap channels and their secrecy capacity.

\begin{example}\label{ex:superact}
\begin{figure}
\centering
	\begin{tikzpicture}[alphabet/.style={draw, rounded corners}]
		\node[alphabet] (Eve) {$\mbf C$};
		\node[alphabet, right = 2cm of Eve] (Alice) {$\mbf A$};
		\node[alphabet, right = 2cm of Alice] (Bob) {$\mbf B$};
		
		\node[below = .2cm of Alice] (a1) {$a_1$};
		\node[below = 1cm of Alice] (a2) {$a_2$};
		\node[below = 1.8cm of Alice] (a3) {$a_3$};
		\node[below = 2.6cm of Alice] (a4) {$a_4$};
		
		\node[below = 1cm of Eve] (c1) {$c_1$};
		\node[below = 1.8cm of Eve] (c2) {$c_2$};
		
		\node[below = .6cm of Bob] (b1) {$b_1$};
		\node[below = 1.4cm of Bob] (b2) {$b_2$};
		\node[below = 2.2cm of Bob] (b3) {$b_3$};
		
		\draw (a1) -- (b1);
		\draw (a1) -- (c1);
		\draw (a2) -- (b1);
		\draw (a2) -- (b2);
		\draw (a2) -- (c1);
		\draw (a2) -- (c2);
		\draw (a3) -- (b2);
		\draw (a3) -- (b3);
		\draw (a3) -- (c1);
		\draw (a3) -- (c2);
		\draw (a4) -- (b3);
		\draw (a4) -- (c2);
	\end{tikzpicture}
	\caption{The channel $(\mbf T_B,\mbf T_C)$ from Example \ref{ex:superact}. A line between $a_i$ and $b_j$ indicates that $b_j\in\mbf T_B(a_i)$, similar for $a_i$ and $c_j$.}\label{fig:superact}
\end{figure}
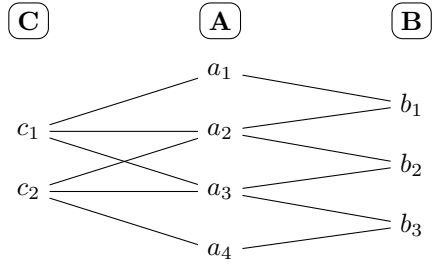

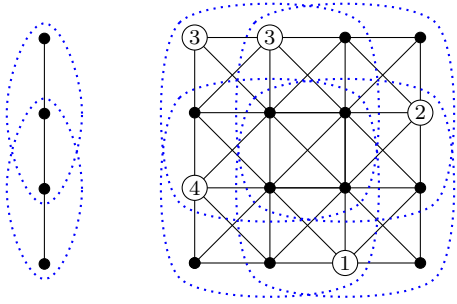
\begin{figure}
\centering
\subfloat{%
\begin{tikzpicture}[every node/.style={draw,circle,fill=black,minimum size=4pt, inner sep=0pt}]%
	\draw (0,0) node (x1) {}
		-- ++(0,1) node (x2) {}
		-- ++(0,1) node (x3) {}
		-- ++(0,1) node (x4) {};
		
	\node[draw=white, fill = white, below=.3cm of x1]  {};
		
	\draw [blue, dotted, thick] plot [smooth cycle, tension = 0.8] coordinates {(0,-.2) (.5,1) (0,2.2) (-.5,1)};
	\draw [blue, dotted, thick] plot [smooth cycle, tension = 0.8] coordinates {(0,.8) (.5,2) (0,3.2) (-.5,2)};
\end{tikzpicture}
}
\qquad
\subfloat{
	\begin{tikzpicture}[every node/.style={draw,circle,fill=black,minimum size=4pt, inner sep=0pt}]
	\draw (1,0) -- (1,3);
	\draw (2,0) -- (2,3);
	\draw (0,1) -- (3,1);
	\draw (0,2) -- (3,2);
	
	\draw (0,2) -- (1,3);
	\draw (0,1) -- (2,3);
	\draw (0,0) -- (3,3);
	\draw (1,0) -- (3,2);
	\draw (2,0) -- (3,1);
	
	\draw (0,2) -- (2,0);
	\draw (0,1) -- (1,0);
	\draw (0,3) -- (3,0);
	\draw (1,3) -- (3,1);
	\draw (2,3) -- (3,2);
	
	\draw (0,0) node (00) {} -- ++(1,0) node {} -- ++(1,0) node[draw, fill=white, shape=circle, minimum size=8pt, inner sep=1pt] {\footnotesize 1} -- ++(1,0) node {}
		-- ++(0,1) node {} -- ++(0,1) node[draw,fill = white, shape=circle, minimum size=8pt, inner sep=1pt] {\footnotesize 2} -- ++(0,1) node {}
		-- ++(-1,0) node {} -- ++(-1,0) node[draw, fill = white, shape=circle, minimum size=8pt, inner sep=1pt] {\footnotesize 3} -- ++(-1,0) node[draw,fill = white, shape=circle, minimum size=8pt, inner sep=1pt] {\footnotesize 3}
		-- ++(0,-1) node {} -- ++(0,-1) node[draw,fill = white, shape=circle, minimum size=8pt, inner sep=1pt] {\footnotesize 4} -- ++(0,-1) node {};
		
	\draw (1,1) node {} -- (2,1) node {} -- (2,2) node {} -- (1,2) node {};

	\draw [blue, dotted, thick] plot [smooth cycle, tension = 0.5] coordinates {(-.2,-.2) (-.2,2.2) (2.2,2.2) (2.2,-.2)};
	\draw [blue, dotted, thick] plot [smooth cycle, tension = 0.5] coordinates {(-.2,.8) (-.2,3.2) (2.2,3.2) (2.2,.8)};
	\draw [blue, dotted, thick] plot [smooth cycle, tension = 0.5] coordinates {(.8,-.2) (.8,2.2) (3.2,2.2) (3.2,-.2)};
	\draw [blue, dotted, thick] plot [smooth cycle, tension = 0.5] coordinates {(.8,.8) (.8,3.2) (3.2,3.2) (3.2,.8)};
\end{tikzpicture}
}
\caption{Left: $\mbf A$ with $G(\mbf T_B)$ and $H(\mbf T_C)$. Right: $\mbf A^2$ with $G(\mbf T_B^2)$ and $H(\mbf T_C^2)$. Vertices connected by a solid black line are connected in $G(\mbf T_B)$ or $G(\mbf T_B^2)$, respectively. Vertices within the boundary of a blue dotted line belong to the same hyperedge of $H(\mbf T_C)$ or $H(\mbf T_C^2)$, respectively.}\label{fig:superactgraphs}
\end{figure}
	Consider the wiretap channel $(\mbf T_B,\mbf T_C)$ from Fig.\ \ref{fig:superact}. $\mbf A$ with $G(\mbf T_B)$ and $H(\mbf T_C)$ is depicted on the left of Fig.\ \ref{fig:superactgraphs}, $\mbf A^2$ with $G(\mbf T_B^2)$ and $H(\mbf T_C^2)$ on its right. It is easy to check that there is no zero-error wiretap $(M,1)$-code for any $M\geq2$. On the other hand, a zero-error wiretap $(4,2)$-code exists by chooosing the codeword sets as indicated in Fig.\ \ref{fig:superactgraphs}. Therefore in the general case, in contrast to the situation in Lemma \ref{lem:TBinj}, there is no easy criterion an uncertain wiretap channel satisfies at blocklength 1 if and only if its zero-error secrecy capacity is positive. 
	
	This behavior of zero-error wiretap codes for general uncertain wiretap channels is remarkable when it is compared to the behavior of zero-error codes for uncertain channels: An uncertain channel $\mbf T$ has $C_0(\mbf T)>0$ if and only if there exists an independent system for $G(\mbf T)$ with size $\geq 2$. Similarly, a stochastic DMC has positive capacity if and only if its blocklength-1 transmission matrix does not have identical rows. For the secrecy capacity of stochastic wiretap channels, there is van Dijk's criterion \cite{vD} for positivity which concerns the blocklength-1 wiretap channel matrix and requires to check a certain function for concavity.
	
	Observe also that in order to obtain a $(4,2)$-code for the above channel, one message $m$ has to be encoded into a set with $\lvert\mbf F(m)\rvert\geq2$. 
\end{example}

A simpler example illustrating the necessity of codes whose encoding sets $\mbf F(m)$ are not all one-element sets is the following.

\begin{example}\label{ex:singletonnec}
	Consider the wiretap channel shown in Fig.\ \ref{fig:singletonnec}. If one only considered codes satisfying $\lvert\mbf F(m)\rvert=1$ for all messages $m$, then the maximal $M$ for which a zero-error wiretap $M$-code exists would be $M=2$, for example $\mbf F=\{\{a_1\},\{a_4\}\}$. $M=4$ is not possible because $\mbf T_B$ can only transmit 3 messages without error. For $M=3$, either $c_1$ or $c_2$ would be generated by only one message. 

On the other hand, if one takes the zero-error wiretap code $\mbf F=\{\{a_1\},\{a_2,a_3\},\{a_4\}\}$, then three messages can be distinguished at the intended receiver's output and every eavesdropper output is reached by two different messages. 
\end{example}

Moreover, examples can be constructed which show the following: If there exists a zero-error wiretap $(M,n)$-code, then it may be necessary to have codes with $\lvert\mbf F(m)\rvert\geq 2$ to also find a zero-error wiretap $(M',n)$-code for every $2\leq M'\leq M$.

\section{Proofs}

This section contains all the proofs. The first two subsections are devoted to the proof of Theorem \ref{thm:main}. Subsection \ref{subsect:prelim} contains the quantizer rule applied by the sensor and some basic lemmas which are needed in the analysis of the transmission scheme to be defined. The transmission scheme is defined and analyzed in Subsection \ref{subsect:transm}. The proofs of Lemmas \ref{lem:esterr} and \ref{lem:divcoeff} which are based on the transmission scheme defined in Subsection \ref{subsect:transm} are done in Subsection \ref{subsect:lemmas}. The proof of Theorem \ref{lem:TBinj} is contained in Subsection \ref{subsect:TBinj}, followed by the proof of Lemma \ref{lem:mittldist} in Subsection \ref{subsect:mittldist}.

\subsection{Proof of Theorem \ref{thm:main}: Preliminaries}\label{subsect:prelim}

The first choice to make is the quantizer used by the sensor. For sufficient generality, we assume the rule \eqref{eq:dynrec}, but that $x(0)\in I_0$ for some real interval $I_0$. Let $M\geq 2\in\mbb N$. Recursively define for $t\geq 1$ and $1\leq m\leq M$
\begin{align}
	[A(t),B(t)]&=\lambda I_{t-1}+\left[-\frac{\Omega}{2},\frac{\Omega}{2}\right],\label{eq:A(n),B(n)}\\
	P_{m,t}&=A(t)+(B(t)-A(t))\left[\frac{m-1}{M},\frac{m}{M}\right],\label{eq:Pmn}\\
	m_{t}&=m\qquad\text{if }x(t)\in P_{m,t},\\
	I_t&=P_{m_t,t}.\label{eq:In}%\\
%	\hat x(t)&=A(t)+\frac{2m_{t}-1}{2M}(B(t)-A(t)).\label{eq:hatx(n)}
\end{align}
In the definition of $m_{t}$, an uncertain mapping is applied to associate $x(t)$ to one of the two possible values if it lies on the boundary between two partition intervals $P_{m,t},P_{m+1,t}$.

For every $t\in\mbb N$, the interval $I_t$ is the set of system states which are possible at time $t$ according to the sequence $(m_i)_{i=t}^n$. The interval $[A(t+1),B(t+1)]$ is the set of states the system could be in at time $t+1$ given that its state at time $t$ is contained in $I_t$. The sets $P_{m,t+1}:1\leq m\leq M$ form an equal-sized partition of $[A(t+1),B(t+1)]$, and $m_{t+1}$ is the index of the partition atom actually containing the system state. Clearly, every path $(x(t))_{t=0}^\infty$ generates an infinite sequence $(m_t)_{t=1}^\infty$.

The next lemma is needed in the analysis of the intended receiver's estimation error and proved by induction over the recursion \eqref{eq:A(n),B(n)}-\eqref{eq:In}.

\begin{lem}\label{thm:estdyn}
	For every $t\in\mbb N$ and $1\leq m\leq M$,
	\begin{align}\label{eq:intlength}
		\lvert P_{m,t}\rvert
		=\begin{cases}
			\left(\frac{\lambda}{M}\right)^t\left(\lvert I_0\rvert-\frac{\Omega}{M-\lambda}\right)+\frac{\Omega}{M-\lambda}&\text{if }\lambda\neq M,\\
			\lvert I_0\rvert+t\frac{\Omega}{M}&\text{if }\lambda=M.
		\end{cases}
	\end{align}
	In particular, $\sup_t\lvert I_t\rvert<\infty$ if and only if $\lambda<M$. In that case
	\[
		\sup_t\lvert I_t\rvert=\max\left\{\lvert I_0\rvert,\frac{\Omega}{M-\lambda}\right\}.
	\]
\end{lem}

\begin{proof}
Write $I_t:=[I_{t,\min},I_{t,\max}]$ for every $t\in\mbb N$. Then note that by \eqref{eq:A(n),B(n)} 
	\begin{equation}\label{eq:AnBnIn}
		[A(t+1),B(t+1)]=\left[\lambda I_{t,\min}-\frac{\Omega}{2},\lambda I_{t,\max}+\frac{\Omega}{2}\right],
	\end{equation}
	which implies that $B(t+1)-A(t+1)=\lambda\lvert I_t\rvert+\Omega$. Hence by \eqref{eq:In} and \eqref{eq:Pmn}
	\begin{equation}\label{eq:Pmnlength}
		\lvert P_{m,t+1}\rvert=\frac{B(t+1)-A(t+1)}{M}=\frac{\lambda}{M}\lvert I_t\rvert+\frac{\Omega}{M}.
	\end{equation}
	Next \eqref{eq:intlength} is established by induction over $t$, using \eqref{eq:Pmnlength}. It is simple if $\lambda=M$. Now assume $\lambda\neq M$. Clearly the statement if true for $t=0$. Using the induction hypothesis and \eqref{eq:Pmnlength}, we then get
	\begin{align*}
		\lvert P_{m,t+1}\rvert&=\frac{\lambda}{M}\lvert I_t\rvert+\frac{\Omega}{M}\\
		&=\frac{\lambda}{M}\left(\left(\frac{\lambda}{M}\right)^t\left(\lvert I_0\rvert-\frac{\Omega}{M-\lambda}\right)+\frac{\Omega}{M-\lambda}\right)+\frac{\Omega}{M}\\
		&=\left(\frac{\lambda}{M}\right)^{t+1}\left(\lvert I_0\rvert-\frac{\Omega}{M-\lambda}\right)+\frac{\Omega}{M-\lambda}\left(\frac{\lambda}{M}+\frac{M-\lambda}{M}\right)\\
		&=\left(\frac{\lambda}{M}\right)^{t+1}\left(\lvert I_0\rvert-\frac{\Omega}{M-\lambda}\right)+\frac{\Omega}{M-\lambda}.
	\end{align*}
	This completes the proof.
\end{proof}

Denote by $\hat x(t)$ the mid point of $I_t$ for $t\geq 0$. For the analysis of the diameter of the set of paths compatible with the eavesdropper's outputs, we first derive a recursion formula for the sequence $(\hat x(t))_{t=1}^\infty$ given a sequence of partition indices $(m_t)_{t=1}^\infty$. 

\begin{lem}\label{thm:mpdyn}
	Let $M\in\mbb N$ and for every $t$ let $1\leq m_t\leq M$. Let $\hat x(t)$ be the mid point of $I_t=P_{m_t,t}$ for $t\in\mbb N$. Define
	\[
		\sigma_t:=\sum_{i=0}^t\left(\frac{\lambda}{M}\right)^t=\frac{M}{M-\lambda}\left(1-\left(\frac{\lambda}{M}\right)^{t+1}\right).
	\] 
	Then for every $t=0,1,2,\ldots$
	\begin{align}
		\hat x(t)&=\lambda^t\left\{\hat x(0)
		-\frac{1}{2}\sum_{i=0}^{t-1}\left(\frac{\Omega\sigma_i}{\lambda^{i+1}}+\frac{\lvert I_0\rvert}{M^i}\right)\left(1-\frac{2m_{i+1}-1}{M}\right)\right\}.
	\end{align}
\end{lem}

\begin{proof}
	For $t\in\mbb N$ and $1\leq m\leq M$ we set $P_{m,t}=[P_{m,t,\min},P_{m,t,\max}]$. By definition, $\hat x(t+1)=P_{m_{t+1},t+1,\min}+\lvert P_{1,t+1}\rvert/2$ (using $\lvert P_{1,t+1}\rvert=\lvert P_{m,t+1}\rvert$ for all $1\leq m\leq M$). Then 
	\begin{align}
		\hat x(t+1)&=A(t+1)+(m_{t+1}-\frac{1}{2})\lvert P_{1,t+1}\rvert&&\text{by \eqref{eq:Pmn}}\notag\\
		&=\lambda P_{m_t,t,\min}-\frac{\Omega}{2}+(m_{t+1}-\frac{1}{2})\lvert P_{1,t+1}\rvert&&\text{by \eqref{eq:A(n),B(n)}}\notag\\
		&=\lambda\hat x(t)-\frac{\lambda\lvert P_{1,t}\rvert}{2}-\frac{\Omega}{2}+(m_{t+1}-\frac{1}{2})\lvert P_{1,t+1}\rvert&&\text{by def. of }\hat x(t)\notag\\
		&=\lambda\hat x(t)-\frac{\lambda\lvert P_{1,t}\rvert}{2}-\frac{\Omega}{2}+(m_{t+1}-\frac{1}{2})\left(\frac{\lambda}{M}\lvert P_{1,t}\rvert+\frac{\Omega}{M}\right)&&\text{by \eqref{eq:Pmnlength}}\notag\\
		&=\lambda\hat x(t)-\frac{\lambda\lvert P_{1,t}\rvert+\Omega}{2}\left(1-\frac{2m_{t+1}-1}{M}\right).\label{eq:hatxrec}
	\end{align}
	By Lemma \ref{thm:estdyn} and \eqref{eq:hatxrec}
	\begin{align}
		\hat x(t+1)&=\lambda\hat x(t)-\left(\frac{\lambda^{t+1}\lvert I_0\rvert}{2M^t}-\frac{\lambda^{t+1}}{2M^t}\frac{\Omega}{M-\lambda}+\frac{\lambda}{2}\frac{\Omega}{M-\lambda}+\frac{\Omega}{2}\right)\left(1-\frac{2m_{t+1}-1}{M}\right)\notag\\
		&=\lambda\hat x(t)-\left(\frac{\lambda^{t+1}\lvert I_0\rvert}{2M^t}-\frac{\Omega}{2}\frac{\lambda^{t+1}-\lambda M^t-(M-\lambda)M^t}{M^t(M-\lambda)}\right)\left(1-\frac{2m_{t+1}-1}{M}\right)\notag\\
		&=\lambda\hat x(t)-\left(\frac{\lambda^{t+1}\lvert I_0\rvert}{2M^t}-\frac{\Omega}{2}\frac{M}{M-\lambda}\frac{\lambda^{t+1}-M^{t+1}}{M^{t+1}}\right)\left(1-\frac{2m_{t+1}-1}{M}\right)\notag\\
		&=\lambda\hat x(t)-\left(\frac{\lambda^{t+1}}{M^t}\frac{\lvert I_0\rvert}{2}+\frac{\Omega}{2}\sigma_t\right)\left(1-\frac{2m_{t+1}-1}{M}\right).\label{eq:step}
	\end{align}
	Now we use induction to prove the claim. It is certainly correct for $t=0$. Assume the claim has been proven for all integers up to $t$. We obtain from \eqref{eq:step}
	\begin{align*}
		\hat x(t+1)&=\lambda^{t+1}\left\{\hat x(0)
		-\frac{1}{2}\sum_{i=0}^{t-1}\left(\frac{\Omega\sigma_i}{\lambda^{i+1}}+\frac{\lvert I_0\rvert}{M^i}\right)\left(1-\frac{2m_{i+1}-1}{M}\right)\right\}\\
		&\qquad\qquad-\frac{\lambda^{t+1}}{2}\left\{\left(\frac{\Omega\sigma_{t}}{\lambda^{t+1}}+\frac{\lvert I_0\rvert}{M^t}\right)\left(1-\frac{2m_{t+1}-1}{M}\right)\right\}\\
		&=\lambda^{t+1}\left\{\hat x(0)
		-\frac{1}{2}\sum_{i=0}^{t}\left(\frac{\Omega\sigma_i}{\lambda^{i+1}}+\frac{\lvert I_0\rvert}{M^i}\right)\left(1-\frac{2m_{i+1}-1}{M}\right)\right\}.
	\end{align*}
	This completes the proof.
\end{proof}

Next assume that we have two systems obeying \eqref{eq:dynrec}. The paths of one of them start in an interval $I_0$ and those of the other in an interval $I_0'$ with $\lvert I_0\rvert=\lvert I'_0\rvert$. The same quantizer rules \eqref{eq:A(n),B(n)}-\eqref{eq:In} are applied for both systems, generating sequences $(m_t)_{t=1}^\infty, (I_t)_{t=0}^\infty$ and $(m_t')_{t=1}^\infty,(I_t')_{t=0}^\infty$, respectively. For $t\geq 0$, denote by $\hat x(t)$ the mid point of $I_t$ and by $\hat x'(t)$ that of $I_t'$. The next two lemmas will be used in the security analysis of the scheme we are going to define. 

\begin{lem}\label{thm:estandsecwith}
	Let $L,M\geq 2$ and for every $t$ let $1\leq m_t'<m_t\leq M$ with $m_t-m_t'\geq L-1$. Then 
	\[
		\liminf_{t\rightarrow\infty}\frac{\hat x(t)-\hat x'(t)}{\lambda^t}\geq\hat x(0)-\hat x'(0)+\frac{L-1}{M-1}\left(\frac{\Omega}{\lambda-1}+\lvert I_0\rvert\right).
	\]
\end{lem}

\begin{proof}
	By Lemma \ref{thm:mpdyn}, for any $n\in\mbb N$,
	\begin{align}
		\hat x(t)-\hat x'(t)&=\lambda^t\biggl(\hat x(0)-\hat x'(0)+\frac{\Omega}{M}\sum_{i=0}^{t-1}\frac{\sigma_i}{\lambda^{i+1}}(m_{i+1}-m'_{i+1})\label{eq:hatdiff1}\\&\qquad\qquad+\lvert I_0\rvert\sum_{i=0}^{t-1}\frac{1}{M^{i+1}}(m_{i+1}-m'_{i+1})\biggr).\notag
	\end{align}
	Observe that 
	\begin{equation}\label{eq:sigmadurch}
		\frac{\sigma_i}{\lambda^{i+1}}=\frac{M}{M-\lambda}\left(\frac{1}{\lambda^{i+1}}-\frac{1}{M^{i+1}}\right)
	\end{equation}
	and recall that $m_t-m'_t\geq L-1$ for every $t$. Hence \eqref{eq:hatdiff1} can be lower-bounded by
	\begin{align*}
		&\lambda^t\biggl(\hat x(0)-\hat x'(0)+\frac{\Omega(L-1)}{M-\lambda}\sum_{i=0}^{t-1}\left(\frac{1}{\lambda^{i+1}}-\frac{1}{M^{i+1}}\right)+\lvert I_0\rvert(L-1)\sum_{i=0}^{t-1}\frac{1}{M^{i+1}}\biggr)\\
		&=\lambda^t\biggl(\hat x(0)-\hat x'(0)+\frac{\Omega(L-1)}{M-\lambda}\left(\frac{1-\lambda^{-t}}{\lambda-1}-\frac{1-M^{-t}}{M-1}\right)+\lvert I_0\rvert(L-1)\frac{1-M^{-t}}{M-1}\biggr).
	\end{align*}
	The theorem is proven once one observes that as $t\rightarrow\infty$,
	\begin{align*}
		&\frac{\Omega(L-1)}{M-\lambda}\left(\frac{1-\lambda^{-t}}{\lambda-1}-\frac{1-M^{-t}}{M-1}\right)+\lvert I_0\rvert(L-1)\frac{1-M^{-t}}{M-1}\\
		&\longrightarrow\frac{\Omega(L-1)}{M-\lambda}\left(\frac{1}{\lambda-1}-\frac{1}{M-1}\right)+\lvert I_0\rvert(L-1)\frac{1}{M-1}\\
		&=\frac{\Omega(L-1)}{M-\lambda}\frac{M-\lambda}{(M-1)(\lambda-1)}+\lvert I_0\rvert(L-1)\frac{1}{M-1}\\
		&=\frac{L-1}{M-1}\left(\frac{\Omega}{\lambda-1}+\lvert I_0\rvert\right).
	\end{align*}
\end{proof}

\begin{lem}\label{thm:estandsecwithout}
	Let $M\in\mbb N$ and for every $t$ let $1\leq m_t,m_t'\leq M$. If 
	\begin{equation}\label{eq:voraussohne}
		\lvert\hat x(0)-\hat x'(0)\rvert>\frac{\Omega}{\lambda-1}+\lvert I_0\rvert,
	\end{equation}
	then for every $t=1,2,\ldots$
	\[
		\liminf_{t\rightarrow\infty}\frac{\lvert\hat x(t)-\hat x'(t)\rvert}{\lambda^t}\geq\lvert\hat x(0)-\hat x'(0)\rvert-\frac{\Omega}{\lambda-1}-\lvert I_0\rvert.
	\]
\end{lem}

\begin{proof}
By Lemma \ref{thm:estdyn}, 
\begin{equation}\label{eq:estdiff}
	\lvert\hat x(t)-\hat x'(t)\rvert=\lambda^t\left\lvert (\hat x(0)-\hat x'(0))+\frac{1}{M}\sum_{i=0}^{t-1}\left(\frac{\Omega\sigma_i}{\lambda^{i+1}}+\frac{\lvert I_0\rvert}{M^i}\right)(m_{i+1}-m'_{i+1})\right\rvert.
\end{equation}
By the triangle inequality, the absolute value term on the right-hand side of \eqref{eq:estdiff} is lower-bounded by
\begin{equation}\label{eq:triang}
	\left\lvert\lvert\hat x(0)-\hat x'(0)\rvert-\biggl\lvert\frac{1}{M}\sum_{i=0}^{t-1}\left(\frac{\Omega\sigma_i}{\lambda^{i+1}}+\frac{\lvert I_0\rvert}{M^i}\right)(m_{i+1}-m'_{i+1})\biggr\rvert\right\rvert.
\end{equation}
Using \eqref{eq:sigmadurch},
\begin{align}
	&\biggl\lvert\frac{1}{M}\sum_{i=0}^{t-1}\left(\frac{\Omega\sigma_i}{\lambda^{i+1}}+\frac{\lvert I_0\rvert}{M^i}\right)(m_i-m'_i)\biggr\rvert\\
	&\leq\frac{M-1}{M}\sum_{i=0}^{t-1}\left(\frac{\Omega M}{M-\lambda}\left(\frac{1}{\lambda^{i+1}}-\frac{1}{M^{i+1}}\right)+\frac{\lvert I_0\rvert}{M^i}\right)\notag\\
	&=(M-1)\left\{\frac{\Omega}{M-\lambda}\frac{1-\lambda^{-t}}{\lambda-1}+\left(\lvert I_0\rvert-\frac{\Omega}{M-\lambda}\right)\frac{1-M^{-t}}{M-1}\right\}\notag\\
	&=\frac{\Omega}{M-\lambda}\frac{M-1}{\lambda-1}\left(1-\frac{1}{\lambda^t}\right)+\left(\lvert I_0\rvert-\frac{\Omega}{M-\lambda}\right)\left(1-\frac{1}{M^t}\right).\label{eq:compl}
\end{align}
As $t$ tends to infinity, \eqref{eq:compl} converges to 
\[
	\frac{\Omega}{M-\lambda}\left(\frac{M-1}{\lambda-1}-1\right)+\lvert I_0\rvert=\frac{\Omega}{\lambda-1}+\lvert I_0\rvert.
\]
This proves the lemma.
\end{proof}

\subsection{Proof of Theorem \ref{thm:main}: Transmission Schemes}\label{subsect:transm}

For any $n\geq 1$, let us introduce the \textit{$n$-sampled system}
\begin{align*}
	x\e{n}(k+1)&=\lambda^nx^{(n)}(k)+w^{(n)}(k),\\
	x\e{n}(0)&=0,
\end{align*}
where $w^{(n)}(k)$ is a nonstochastic disturbance in the range $[-\tilde\Omega^{(n)}/2,\tilde\Omega^{(n)}/2]$ (cf.\ \eqref{eq:finhorizbound}). The $n$-sampled system describes the system \eqref{eq:dynsyst} at the points $0,n,2n,\ldots$

Let us first assume that $C_0(\mbf T_B,\mbf T_C)\leq\log\lambda$. Choose $n_1,M_1$ such that $2\leq M_1<\lambda^{n_1}$ and $M_1\leq N_{(\mbf T_B,\mbf T_C)}(n_1)$ and choose $n_2$ such that $M_2:=N_{\mbf T_B}(n_2)>\lambda^{n_2}$. Let $L\geq 2$ be chosen such that there exists a zero-error wiretap $(M_1,L,n_1)$-code $\mbf F$ and let $\mbf G$ be a zero-error $(M_2,n_2)$-code. 

We define a transmission scheme as follows: We do the construction \eqref{eq:A(n),B(n)}-\eqref{eq:In} for the $n_1$-sampled system with $M$ replaced by $M_1$ and with $I_0=\{0\}$, thus obtaining $A\e{n_1}(k),B\e{n_1}(k)$, $P_{m,k}\e{n_1},m_k,I_k\e{n_1}$ (omitting the superscript $(n_1)$ at $m_k$). For some $K\in\mbb N$ to be chosen later and $1\leq k\leq K$, we set
\[
	f_k(x(0),\ldots,x(kn_1))=\mbf F(m_k)
\]
The intended receiver uses the mid point $\hat x(kn_1)$ of $P_{m_k,k}\e{n_1}$ as estimate of $x(kn_1)$. For $k>K$, we first define $A\e{n_2}(k-K),B\e{n_2}(k-K),P_{m,k-K}\e{n_2},m_{k-K},I_{k-K}\e{n_2}$ as in \eqref{eq:A(n),B(n)}-\eqref{eq:In} but with $I_0=P_{m_K,K}\e{n_1}$ (and again omitting the superscript $(n_2)$ at $m_{k-K}$). We then set
\[
	f_k(x(0),\ldots,x(Kn_1+(k-K)n_2)=\mbf G(m_{k-K})
\]
Decoding/estimating goes as in the first $K$ steps.

As $M_2>\lambda^{n_2}$ it is clear that the estimation error for the intended receiver at decoding times $(t_k)_{k=1}^\infty$ equals
\begin{equation}\label{eq:max1}
	\max\left\{\max_{1\leq k\leq K}\lvert P\e{n_1}_{m_k,k}\rvert,\sup_{k}\lvert P_{m_k,k-K}\e{n_2}\rvert\right\}<\infty.
\end{equation}
More precisely, by Lemma \ref{thm:estdyn}, the maximum in the curly brackets in \eqref{eq:max1} equals
\begin{equation}\label{eq:PmKlength}
	\lvert P_{m_K,K}\rvert=\left(\left(\frac{\lambda^{n_1}}{M_1}-1\right)^K\right)\frac{\Omega}{\lambda-1}\frac{\lambda^{n_1}-1}{\lambda^{n_1}-M_1}
\end{equation}
and the supremum inside the curly brackets in \eqref{eq:max1} equals the maximum of \eqref{eq:PmKlength} and
\begin{equation}\label{eq:anderelength}
	\frac{\Omega}{\lambda-1}\frac{\lambda^{n_2}-1}{M_2-\lambda^{n_2}}.
\end{equation}
Thus the intended receiver's estimation error is bounded at decoding times. In between, it can only grow finitely, so the total estimation error is bounded.

To prove security of the transmission scheme defined above, fix an $\varepsilon>0$. Now assume the eavesdropper receives a channel output sequence $(c_t)_{t=1}^\infty$. Lemma \ref{thm:estandsecwith} implies the existence of paths $(x(t))_{t=0}^\infty,(x'(t))_{t=0}^\infty$ such that for sufficiently large $K$, the estimates at time $Kn_1$ have distance 
\begin{equation}\label{eq:coeffafterK}
	\hat x(Kn_1)-\hat x'(Kn_1)\geq\lambda^{Kn_1}\left(\frac{L-1}{M_1-1}\frac{\Omega}{\lambda-1}-\varepsilon\right)
\end{equation}
(note that here, Lemma \ref{thm:estandsecwith} has to be applied with $\lvert I_0\rvert=0$ and $\hat x(0)=\hat x'(0)$). By choosing $K$ even larger if necessary, \eqref{eq:voraussohne} is satisfied with its left-hand side replaced by $\lvert\hat x(Kn)-\hat x'(Kn)\rvert$ and the right-hand side by 
\begin{equation}\label{eq:brbr}
	\frac{\tilde\Omega\e{n_2}}{\lambda^{n_2}-1}+\lvert P_{m_K,K}\rvert.
\end{equation}
This can be seen by applying \eqref{eq:coeffafterK} and by using \eqref{eq:PmKlength} to show that \eqref{eq:brbr} equals
\[
	\frac{\tilde\Omega^{(n_2)}}{\lambda^{n_2}-1}+\lvert P_{m,K}^{(n_1)}\rvert=\frac{\Omega}{\lambda-1}\left(1+\left(\left(\frac{\lambda^{n_1}}{M_1}\right)^K-1\right)\frac{\lambda^{n_1}-1}{M_1-\lambda^{n_1}}\right).
\]
One can thus apply Lemma \ref{thm:estandsecwithout} to find that for sufficiently large $k$ (and after enlarging $K$ again if necessary), the distance between $\hat x(kn_2)$ and $\hat x'(kn_2)$ is lower-bounded by
\begin{equation}\label{eq:divcoeffrawersterfall}
	\lambda^{Kn_1+(k-K)n_2}\frac{\Omega}{\lambda-1}\left(\frac{L-1}{M_1-1}-\frac{1}{\lambda^{Kn_1}}-\left(\frac{1}{M_1^{K}}-\frac{1}{\lambda^{Kn_1}}\right)\frac{\lambda^{n_1}-1}{M_1-\lambda^{n_1}}-2\varepsilon\frac{\lambda-1}{\Omega}\right).
\end{equation}
This tends to infinity as $k\rightarrow\infty$ and thus proves that the transmission scheme defined satisfies security. We have thus proved that there exists a reliable and secure transmission scheme in the case $C_0(\mbf T_B)>\log\lambda$ and $0<C_0(\mbf T_B,\mbf T_C)\leq\log\lambda$. 

Next we treat the case $C_0(\mbf T_B,\mbf T_C)>\log\lambda$. The construction is simpler than the previous case, as it applies the same zero-error wiretap code in every time step. Choose $n$ such that $M:=N_{(\mbf T_B,\mbf T_C)}(n)>\lambda^n$. Let $L\geq 2$ be chosen such that there exists a zero-error wiretap $(M,L,n)$-code $\mbf F$. 

We define a transmission scheme as follows: The construction \eqref{eq:A(n),B(n)}-\eqref{eq:In} is done for the $n$-sampled system with $I_0=\{0\}$ and thus obtain $A\e n(k),B\e n(k), P_{m,k}\e n,m_k, I_k\e n$ (omitting the superscript $(n)$ at $m_k$). We then set 
\[
	f_k(x(0),\ldots,x(kn))=\mbf F(m_k).
\]
Again, the intended receiver uses the mid point of $P_{m_k,k}\e n$ as estimate of $x(kn)$. 

By Lemma \ref{thm:estdyn}, the estimation error at times $0,n,2n,\ldots$ is bounded by 
\begin{equation}\label{eq:esterrzweiterfall}
	\frac{\Omega}{\lambda-1}\frac{\lambda^n-1}{M-\lambda^n}.
\end{equation}
Between these times, the error grows, but stays bounded. Hence the total estimation error is bounded, so the above transmission scheme is reliable.

For security, we apply Lemma \ref{thm:estandsecwith} and find that for any $\varepsilon>0$, any eavsdropper sequence $(c_t)_{t=1}^\infty)$ and sufficiently large $k$, there exist paths $(x(t))_{t=0}^\infty,(x'(t))_{t=0}^\infty\in\mbf R_C((c_t)_{t=1}^\infty)$ such that
\begin{equation}\label{eq:divcoeffrawzweiterfall}
	\hat x(kn)-\hat x'(kn)\geq\lambda^{kn}\left(\frac{\Omega}{\lambda-1}\frac{L-1}{M-1}-\varepsilon\right).
\end{equation}
Thus the transmission scheme also is secure. Altogether, this proves Theorem \ref{thm:main}.

\subsection{Proofs of Lemmas \ref{lem:esterr} and \ref{lem:divcoeff}}\label{subsect:lemmas}

We distinguish the cases $C_0(\mbf T_B,\mbf T_C)\leq\log\lambda$ and $C_0(\mbf T_B,\mbf T_C)>\log\lambda$ and treat both lemmas for each case at once. 

Let us start with the case $C_0(\mbf T_B,\mbf T_C)\leq\log\lambda$. The maximal estimation error at decoding times $0,n_1,\ldots,Kn_1,Kn_1+n_2,Kn_1+2n_2,\ldots$ is given by \eqref{eq:max1}, i.~e.\ the maximum of \eqref{eq:PmKlength} and \eqref{eq:anderelength}. The error \eqref{eq:PmKlength} is obtained at time $Kn_1$, whereas \eqref{eq:anderelength} is the asymptotic error as $k\rightarrow\infty$. By choosing $n_2$ sufficiently large, this asymptotic error can be made arbitrarily small by choice of $M_2$. Thus for any $\varepsilon>0$ and for sufficiently large $n_2=n_2(\varepsilon)$, we obtain
\[
	\lvert x(Kn_1+(k-K)n_2)-\hat x(Kn_1+(k-K)n_2)\rvert\leq\varepsilon.
\]
This proves Lemma \ref{lem:esterr} for the case $C_0(\mbf T_B,\mbf T_C)\leq\log\lambda$. 

To also show Lemma \ref{lem:divcoeff}, we just need to have a look at \eqref{eq:divcoeffrawersterfall}. First we choose $n_1$ so large that 
\[
	\frac{L-1}{M_1-1}\geq\sup_n\Delta_{(\mbf T_B,\mbf T_C)}(n)-\varepsilon.
\]
Thus the term in the outer brackets in \eqref{eq:divcoeffrawersterfall} is lower bounded by 
\[
	\sup_n\Delta_{(\mbf T_B,\mbf T_C)}(n)-\frac{1}{\lambda^{Kn_1}}-\left(\frac{1}{M_1^{Kn_1}}-\frac{1}{\lambda^{Kn_1}}\right)\frac{\lambda^{n_1}-1}{M_1-\lambda^{n_1}}-\varepsilon\left(1+2\frac{\lambda-1}{\Omega}\right)
\]
Next with sufficiently large $K$, it can be ensured that 
\begin{equation}\label{eq:formel1}
	\frac{1}{\lambda^{Kn_1}}+\left(\frac{1}{M_1^{Kn_1}}-\frac{1}{\lambda^{Kn_1}}\right)\frac{\lambda^{n_1}-1}{M_1-\lambda^{n_1}}+\varepsilon\left(1+2\frac{\lambda-1}{\Omega}\right)\leq2\varepsilon\left(1+\frac{\lambda-1}{\Omega}\right).
\end{equation}
Recall that $\varepsilon$ depends on $K$ and can be made arbitrarily small by enlarging $K$. Hence the term on the right-hand side of \eqref{eq:formel1} can be made arbitrarily small. This proves Lemma \ref{lem:divcoeff} for the case $C_0(\mbf T_B,\mbf T_C)\leq\log\lambda$. 

We next prove Lemmas \ref{lem:esterr} and \ref{lem:divcoeff} to also hold for the case $C_0(\mbf T_B,\mbf T_C)>\log\lambda$. By \eqref{eq:esterrzweiterfall} and the choice of $M$, the estimation error at decoding times can be made arbitrarily small by choosing $n$ sufficiently large. Note that this gives the claimed upper bound on the \textit{supremum} of all estimation errors at decoding times. This proves Lemma \ref{lem:esterr}. The proof of Lemma \ref{lem:divcoeff} is simple as well because of \eqref{eq:divcoeffrawzweiterfall}.

\subsection{Proof of Theorem \ref{lem:TBinj}}\label{subsect:TBinj}

For the proof of Theorem \ref{lem:TBinj}, observe that one can restrict attention to codes with $\lvert\mbf F(m)\rvert=1$ because no vertices are connected in $G(\mbf T_B^n)$ for any $n$. At blocklength $n$, the only question will be how many elements of $\mbf A^n$ can be used as codewords. We write $a_1^n:=(a_1,\ldots,a_n)$ for elements of $\mbf A^n$ and use analogous notation for elements $c_1^n\in\mbf C^n$. It has to be ensured that the eavesdropper cannot infer the codeword $a_1^n$, and thus the message, from its received $c_1^n\in\mbf C^n$. 

To formalize this, we introduce the notion of ``subhypergraph" of a hypergraph. Given a hypergraph $H$ with vertex set $\mbf V$ and hyperedge set $\mc E_H$, we call $\tilde H$ a \textit{subhypergraph of $H$} if the vertex set $\tilde{\mbf V}$ of $\tilde H$ is a subset of $\mbf V$ and if each of the hyperedges of $\tilde H$ has the form $\tilde e=e\cap\tilde{\mbf V}$ for some $e\in\mc E_H$ (the empty set is not allowed as hyperedge). Obviously, the subhypergraph $\tilde H$ is uniquely determined by $\tilde{\mbf V}$ and we denote it by $H\rvert_{\tilde{\mbf V}}$.

Denote by $\mbf T_C^n\rvert_{\mbf V}$ the channel $\mbf T_C^n$ restricted to inputs from $\mbf V\subset\mbf A^n$ and observe that the hypergraph $H(\mbf T_C^n\rvert_{\mbf V})$ is given by the subhypergraph $H(\mbf T_C^n)\rvert_{\mbf V}$ of $H(\mbf T_C^n)$. We can thus formulate our problem by saying that we have to find a large subhypergraph $H\e n$ of $H(\mbf T_C^n)$ which does not contain any hyperedge of cardinality 1.
 
This subhypergraph is found in several consecutive steps. We set $H(\mbf T_C^n)=:H\e n(0)$. First we eliminate from the possible channel input alphabet $\mbf A^n$ all elements $a_1^n$ which can be uniquely determined by the eavesdropper, i.~e.\ all $a_1^n$ such that $\{a_1^n\}$ is a hyperedge of $H\e n(0)$. If we write
\[
	\mbf A_1\e n(1):=\{a_1^n\in\mbf A^n:\{a_1^n\}\text{ is a hyperedge of }H\e n(0)\},
\]
and $\mbf A_2\e n(1):=\mbf A^n\setminus\mbf A_1\e n(1)$, we thus obtain the subhypergraph $H\e n(1):=H\e n(0)\rvert_{\mbf A_2\e n(1)}$ of $H\e n(0)$. 

Now $H\e n(1)$ may again contain hyperedges with cardinality 1: precisely those which have the form $e'=e\cap\mbf A_2\e n(1)$ for a hyperedge $e$ of $H\e n(0)$ which equals $e=\{a_1^n,\tilde a_1^n\}$ for some $a_1^n\in\mbf A_1\e n(1)$ and $\tilde a_1^n\in\mbf A_2\e n(1)$. Thus again eliminating those elements $a_1^n$ from $\mbf A_2\e n(1)$ where $\{a_1^n\}$ is a hyperedge of $H\e n(1)$, one arrives at a subhypergraph $H\e n(2)$, and so on. 

Formally, with $\mbf A_2\e n(0):=\mbf A^n$, we set for $s\geq 1$
\begin{align*}
	\mbf A_1\e n(s)&:=\{a_1^n\in\mbf A_2\e n(s-1):\{a_1^n\}\text{ is a hyperedge in }H\e n(s-1)\},\\
	\mbf A_2\e n(s)&:=\mbf A^n\setminus\mbf A_1\e n(s),\\
	H\e n(s)&:=H\e n(s-1)\rvert_{\mbf A_2\e n(s)}.
\end{align*}
After a finite number $S\e n$ of steps we arrive at a hypergraph $H\e n:=H\e n(S\e n)$ which is either empty or does not contain any hyperedge of cardinality 1. We denote the vertex set of $H\e n$ by $\mbf A_2\e n$ and define $\mbf A_1\e n:=\mbf A^n\setminus\mbf A_2\e n$. Observe that 
\begin{align}
	&\mbf A^n=\mbf A_2\e n(0)\supset\mbf A_2\e n(1)\supset\ldots\supset\mbf A_2\e n(S\e n)=\mbf A_2\e n,\notag\\
	&\mbf A\e n(1)\subset\ldots\subset\mbf A\e n(S\e n)=\mbf A_1\e n.\label{eq:A1sub}
\end{align}

The main step now is to prove $\mbf A_1\e n\subset(\mbf A_1\e 1)^n$ for every $n\geq 1$. Due to \eqref{eq:A1sub}, this is implied by
\begin{equation}\label{eq:claim}
	\mbf A_1\e n(s)\subset(\mbf A_1\e 1)^n\qquad\text{ for every }1\leq s\leq S\e n.
\end{equation}
For $n=1$ nothing has to be proved. For every $n\geq 2$ we prove \eqref{eq:claim} by induction over $s$.

Let $n\geq 2$ and $s=1$. If $a_1^n\in\mbf A_1\e n(1)$, then $\{a_1^n\}$ is a hyperedge in $H\e n(0)$. As $H\e n(0)=H\e1(0)^n$, i.~e.\ $H\e n(0)$ is the $n$-fold square product of $H\e 1(0)$ with itself, this is only possible if $a_i\in\mbf A_1\e1(1)\subset\mbf A_1\e1$ for all $1\leq i\leq n$. 

Now assume \eqref{eq:claim} is proven for all $1\leq\sigma\leq s$. Let $a_1^n\in\mbf A_1\e n(s+1)$, so that $\{a_1^n\}$ is a hyperedge in $H\e n(s)$. This implies that there exists a hyperedge $e\e n=\{a_1^n,a_{1,2}^n,\ldots,a_{1,\mu}^n\}$ in $H\e n(0)$ such that for every $2\leq\nu\leq\mu$ there exists a $1\leq\sigma_\nu\leq s$ such that $a_{1,\nu}^n\in\mbf A_1\e n(\sigma_\nu)$. By the induction hypothesis, $a_{1,\nu}^n\in(\mbf A_1\e1)^n$ for every $\nu$. 

Suppose $a_1^n\notin(\mbf A_1\e1)^n$. Then $a_i\notin \mbf A_1\e 1$ for some $1\leq i\leq n$, so $a_i\in\mbf A_2\e1$. Hence for every hyperedge $e$ of $H\e1(0)$ containing $a_i$ there is an $a_e\in\mbf A_2\e1$ not equal to $a_i$ such that both $a_i$ and $a_e$ are contained in $e$. 

Let $\{a_1^n,\tilde a_{1,2}^n,\ldots,\tilde a_{1,\tilde\mu}^n\}$ be any hyperedge in $H\e n(0)$ containing $a_1^n$. As $H\e n(0)$ is the $n$-fold square product of $H\e1(0)$, there must be a $2\leq\tilde\nu\leq\tilde\mu$ such that the $i$-th component of $\tilde a_{1,\tilde\nu}^n$ equals $a_e$ for one of the hyperedges $e$ of $H\e1(0)$ containing $a_i$. In particular, $a_{1,\tilde\nu}^n\notin(\mbf A_1\e1)^n$. However, this contradicts the existence of the hyperedge $e\e n=\{a_1^n,a_{1,2}^n,\ldots,a_{1,\mu}^n\}$ in $H\e n(0)$ which apart from $a_1^n$ only contains elements of $(\mbf A_1\e 1)^n$.

This proves the claim \eqref{eq:claim}, in particular $\mbf A_1\e n\subset(\mbf A_1\e 1)^n$. We therefore find that for every $n\in\mbb N$, the number of messages that can be sent securely equals
\[
	N_{(\mbf T_B,\mbf T_C)}(n)=\lvert\mbf A^n\rvert-\lvert\mbf A_1\e n\rvert\geq\lvert\mbf A\rvert^n-\lvert\mbf A_1\e 1\rvert^n.
\]
If $\mbf A_1\e1$ is a strict subset of $\mbf A$, then 
\[
	C_0(\mbf T_B,\mbf T_C)=\lim_{n\rightarrow\infty}\frac{\log N_{(\mbf T_B,\mbf T_C)}(n)}{n}=\log\lvert\mbf A\rvert.
\]
Otherwise, $C_0(\mbf T_B,\mbf T_C)$ obviously equals 0. This proves Theorem \ref{lem:TBinj}.

\subsection{Proof of Lemma \ref{lem:mittldist}}\label{subsect:mittldist}

The proof of Lemma \ref{lem:mittldist} is based on the fact that the labelling of the encoding sets of an $M$-code $\mbf F$ is arbitrary. Let $\mbf F$ be any $(M,L,n)$-code. $k$-fold concatenation of $\mbf F$ with itself gives a $(M^k,L\e k,kn)$-code $\mbf F^k$. We show that the encoding sets of $\mbf F^k$ can be labelled in such a way that 
\begin{equation}\label{eq:k--dist}
	L\e k=\frac{M^k-1}{M-1}L
\end{equation}
is possible. The idea is to order the messages $k$-tuples $(m_1,\ldots,m_k)$ lexicographically. We define this recursively: For $k=2$, the message pair $(m_1,m_2)$ gets the label
\[
	l\e2(m_1,m_2)=M(m_1-1)+m_2.
\]
For $k\geq 2$ we set
\[
	l\e{k+1}(m_1,\ldots,m_{k+1}):=M(l\e k(m_1,\ldots,m_k)-1)+m_{k+1}.
\]
It is easy to check that the range of values of $l\e k$ is $\{1,\ldots,M^k\}$. 

For the concatenated code, we label the coding set $\mbf F(m_1)\times\cdots\times\mbf F(m_k)$ with $l\e k(m_1,\ldots,m_k)$. Let $(c_1,\ldots,c_{kn})$ be an eavesdropper output sequence. As $\mbf F$ is an $(M,L,n)$-code, for every $1\leq i\leq n$ there are messages $m_i,m_i'$ satisfying $m_i-m_i'\geq L-1$ such that $(c_{(i-1)n+1},\ldots,c_{in})$ is generated by both $m_i$ and $m_i'$. It is easy to show by induction that the distance of $(m_1,\ldots,m_n)$ and $(m_1',\ldots,m_n')$ according to the labelling function $l\e k$ is 
\[
	l\e k(m_1,\ldots,m_n)-l\e k(m_1',\ldots,m'_n)\geq\frac{M^k-1}{M-1}L.
\]

Observe now that 
\[
	\frac{L\e k-1}{M^k-1}\longrightarrow\frac{L}{M-1}
\]
from below as $k\rightarrow\infty$. Thus every ratio $(L-1)/(M-1)$ can be improved by enlarging the blocklength, which proves the claim of Lemma \ref{lem:mittldist}.

% Generated by IEEEtran.bst, version: 1.13 (2008/09/30)

\end{document}